 \documentclass[smallabstract,smallcaptions]{dccpaper}

\usepackage{amsthm}

\usepackage{epsfig}
\usepackage{citesort}
\usepackage{amsmath}
\usepackage{amssymb}
\usepackage{color}
\usepackage{url}
\usepackage{titlesec}
\usepackage{thm-restate}
\usepackage{tikz}
\usepackage[edges]{forest}

\newlength{\figurewidth}
\newlength{\smallfigurewidth}

\DeclareMathOperator{\argmax}{argmax}

\DeclareMathOperator{\polylog}{polylog}

\DeclareMathOperator{\NCD}{NCD}

\let\epsilon\varepsilon
\let\eps\varepsilon

\usepackage{dsfont}

\usepackage{xspace}
\usepackage{xcolor}

\usepackage[linesnumbered,ruled,vlined,boxed]{algorithm2e}
\newlength{\commentWidth}
\let\oldnl\nl
\newcommand{\nonl}{\renewcommand{\nl}{\let\nl\oldnl}}
\DontPrintSemicolon
\SetKwInOut{Input}{Input}\SetKwInOut{Output}{Output}

\newtheorem{theorem}{Theorem}[section]
\newtheorem{lemma}[theorem]{Lemma}
\newtheorem{definition}[theorem]{Definition}
\newtheorem{corollary}[theorem]{Corollary}
\newtheorem{observation}[theorem]{Observation}

\begin{document}

\title
{
\large
\textbf{Sketching and Streaming for Dictionary Compression}
}

\author{%
Ruben Becker$^{\ast}$, Matteo Canton$^{\dag}$, Davide Cenzato$^{\ast}$, \\ Sung-Hwan Kim$^{\ast}$, Bojana Kodric$^{\ast}$,
and Nicola Prezza$^{\ast}$\\[0.5em]
{\small\begin{minipage}{\linewidth}\begin{center}
\begin{tabular}{ccc}
$^{\ast}$Ca' Foscari University of Venice & \hspace*{0.5in} & $^{\dag}$University of Udine \\
Via Torino 155 && Via delle Scienze 206 \\
30172 Venezia, Italy && 33100 Udine, Italy\\
\url{firstname.lastname@unive.it} && \url{canton.matteo@spes.uniud.it }
\end{tabular}
\end{center}\end{minipage}}
}

\maketitle
\thispagestyle{empty}

\bigskip

\begin{abstract}
    We initiate the study of sub-linear sketching and streaming techniques for estimating the output size of common dictionary compressors such as Lempel-Ziv '77, the run-length Burrows-Wheeler transform, and grammar compression. To this end, we focus on a measure that has recently gained much attention in the information-theoretic community and which approximates up to a polylogarithmic multiplicative factor the output sizes of those compressors: the normalized substring complexity function $\delta$. 
    As a matter of fact, $\delta$ itself is a very accurate measure of compressibility: it is monotone under concatenation, invariant under reversals and alphabet permutations, sub-additive, and asymptotically tight (in terms of worst-case entropy) for representing strings, up to polylogarithmic factors.

    We present a data sketch of $O(\epsilon^{-3}\log n + \epsilon^{-1}\log^2 n)$ words that allows computing a multiplicative $(1\pm \epsilon)$-approximation of $\delta$ with high probability, where $n$ is the string length. 
    The sketches of two strings $S_1,S_2$ can be merged in $O(\epsilon^{-1}\log^2 n)$ time to yield the sketch of $\{S_1,S_2\}$, speeding up the computation of \emph{Normalized Compression Distances} (NCD).
    If random access is available on the input, our sketch can be updated in $O(\epsilon^{-1}\log^2 n)$ time for each character right-extension of the string. 
    This yields a polylogarithmic-space algorithm for approximating $\delta$,  improving exponentially over the working space of the state-of-the-art algorithms running in nearly-linear time.
    Motivated by the fact that random access is not always available on the input data, we then 
    present a streaming algorithm computing our sketch in $O(\sqrt n \cdot \log n)$ working space and $O(\epsilon^{-1}\log^2 n)$ worst-case delay per character.
    We show that an implementation of our streaming algorithm can estimate $\delta$ on a dataset of 189GB with a throughput of 203MB per minute while using only 5MB of RAM, and that our sketch speeds up the computation of all-pairs NCD distances by one order of magnitude, with applications to phylogenetic tree reconstruction.
\end{abstract}

\section{Introduction}
Sketching techniques allow to summarize in sub-linear space information on big datasets, enabling the approximation of useful statistics such as high-order moments \cite{alon1996space}, norms \cite{Johnson1984ExtensionsOL}, and frequencies \cite{MisraG82} (to name a few). Additionally, most data sketches can be computed on data streams in sub-linear space, making them attractive in big data scenarios.
In this paper, we consider data sketches summarizing the \emph{information content} of a string as approximated by data compression techniques. 
Previous research on this problem has focused on \emph{empirical entropy}. 
Chakrabarti et al.\ \cite{ChakrabartiEntropyStream} showed that the zero-order empirical entropy $H_0$ of a data stream can be efficiently approximated up to a multiplicative $(1 + \epsilon)$-factor in poly-logarithmic space, but any multiplicative approximation of the $k$-th order entropy $H_k$ requires nearly-linear space for $k \geq 1$. 
In addition to this fact, it is well-known that $H_k$ is a weak measure when the dataset is highly repetitive \cite{KREFT2013115}.
As extensively shown in the literature (see, for example, the survey by Navarro \cite{Navarro21a}), dictionary compression measures such as the number
$z$ of phrases of the Lempel-Ziv'77 factorization (used by \texttt{winzip}, \texttt{7-zip}, \texttt{gzip}, \texttt{xz}), the number $r$ of equal-letter runs in the Burrows-Wheeler transform (used by \texttt{bzip-2}), and the size $g$ of a smallest context-free grammar generating (only) the text, are exempt from such a limitation.
The information-theoretic quality of these measures is strengthened by the fact that \emph{Normalized Compression Distances} based on dictionary compressors yield very precise notions of string similarity \cite{cilibrasi2005clustering}. Sketching and streaming techniques for such measures would thus speed up tasks such as the computation of all-pairs similarities when the underlying metric is based on data compression (useful, for example, in the computation of phylogenetic trees \cite{cilibrasi2005clustering}). 

Motivated by the above considerations, in this paper we present the first sub-linear-space sketching and streaming techniques for estimating the output sizes of dictionary compressors.  
This result is obtained by describing a data sketch yielding a  $(1\pm \epsilon)$-approximation of the \emph{normalized substring complexity} $\delta = \max_{k\geq 1}\{d_k/k\}$, where $d_k$ is the number of distinct length-$k$ substrings of the string, a measure introduced by Raskhodnikova et al.\ in~\cite{RaskhodnikovaRRS13}.
As shown by Kociumaka et al.\ \cite{KociumakaNP23} and Kempa and Kociumaka~\cite{KempaK20}, any of the above dictionary compression measures is lower-bounded by $\delta$ and upper-bounded by $\delta(\log n)^c$,  where  $n$ is the string's length and $c$ is an opportune constant depending on the compressor.
Even better, Bonnie et al.\ in \cite{bonnie2023dandd} experimentally showed that $\delta$, $z$, and $r$ (normalized to the interval $[0,1]$) are almost indistinguishable on collections of genomic data.
As a matter of fact, $\delta$ is known to be an even more  accurate information measure than $z$, $r$, and $g$: it is monotone under string concatenation, invariant under reversals and alphabet permutations, sub-additive, and asymptotically tight (in terms of worst-case entropy) for representing strings, up to polylogarithmic factors \cite{KociumakaNP23}. None of the measures $z,r,g$ possesses \emph{simultaneously} all of these properties.

\vspace{-10pt}
\paragraph{Overview of the paper.}

After providing all necessary definitions in Section \ref{sec:preliminaries}, in Section \ref{sec:delta and NCD} we prove new properties of the normalized substring complexity $\delta$ and of the Normalized Compression Distance~\cite{cilibrasi2005clustering} $\NCD_{\delta}$ based on $\delta$. In particular, we show that $\delta$ is perfectly sub-additive, that $\NCD_{\delta}(x,y)$ always lies in $[0,1]$ (according to \cite{MingSimilarityMetric2004}, this is an indicator that $\delta$ is a compressibility measure of good quality), and that $\NCD_{\tilde\delta}(x,y)$ is an additive $\Theta(\epsilon)$-approximation of $\NCD_{\delta}(x,y)$ if $\tilde \delta$ is a multiplicative $(1\pm\epsilon)$-approximation of $\delta$. This motivates designing data sketches for $\delta$, a problem that we solve in Section \ref{sec:sketch}. 
Our sketch is based on the observation (already noted in \cite{BernardiniFGP20} for the particular case $\epsilon=1$) that $\max_{i\geq 0}\{d_{\lceil(1+\epsilon)^i\rceil}/\lceil(1+\epsilon)^i\rceil\}$ is a $(1-\Theta(\epsilon))$-approximation of $\delta$. 
We approximate $d_k$, for each sampled length $k = \lceil (1+\epsilon)^i \rceil$, by keeping a count-distinct sketch \cite{KaneOptimalCountDistinct2010} for the subset of distinct (Rabin's fingerprints \cite{rabin1981fingerprinting} of the) length-$k$ substrings.
Our sketch uses space polynomial in $\epsilon^{-1}\log n$ and supports updates and queries (returning a $(1\pm \epsilon)$ approximation of $\delta$), in $O(\epsilon^{-1}\log^2 n)$ time. 
The sketches of two strings $S_1$ and $S_2$ can moreover be merged in $O(\epsilon^{-1}\log^2 n)$ time to obtain the sketch of $\{S_1, S_2\}$, from which one can compute an additive $\epsilon$-approximation of $\NCD_{\delta}(S_1, S_2)$.
In Section \ref{sec:streaming} we show how to compute our sketch in sub-linear space on an input stream of length $n$. The main difficulty in achieving sub-linear space is that, in order to compute the Rabin's fingerprints of the stream's length-$k$ substrings, we need random access to the $k$-th most recent stream's character. Since the largest $k$ for which we need to compute $d_k$ is linear in $n$,  storing the most recent $k$ characters would require $\Theta(n)$ working space. 
Our solution relies on the observation that, if $\hat k = \argmax_{k\ge 1}\{d_k/k\}$ is small, then we can afford keeping a sliding window of the last $\hat k$ stream's characters. If, on the other hand, $\hat k$ is large, then the stream is highly repetitive so we can compress it in small space while supporting bookmarked access to its characters.
We conclude in Section \ref{sec:exp} with experimental results. Complete proofs can be found in the full version \cite{becker2023sketching}.

\vspace{-10pt}
\paragraph{Related work.} Bonnie et al.\ \cite{bonnie2023dandd} have already observed that $d_k$ can be efficiently estimated by employing count-distinct sketches, and that this can yield an heuristic algorithm for estimating $\delta$. Their strategy relies on estimating $d_k/k$ for increasing values of $k$, until a local maximum is found. While this strategy works well in practice because, as they showed, $\hat k = \argmax_{k\ge 1}\{d_k/k\}$ tends to be a very small number, on particular strings (for example, Thue-Morse) $\hat k$ is of the order of $\Theta(n)$ and, as a result, computing all the sketches for $d_k$ requires linear space and quadratic processing time in the worst case.
Moreover, local maxima of $d_k/k$ do not always coincide with the global maximum, so this strategy does not yield any provable approximation of $\delta$.
We are not aware of other works in the literature describing data sketches for estimating the output sizes of \emph{dictionary} compressors (the literature on estimating empirical entropy is, on the other end, much richer: see \cite{ChakrabartiEntropyStream}  
and references therein). 
Our results can be viewed also as a space-efficient way to approximate measure $\delta$.
Christiansen et al.~\cite{ChristiansenEKN21} showed how to compute $\delta$ for a given string $T$ in linear time and space. Recently, Bernardini et al.~\cite{BernardiniFGP20}  provided space-time trade-offs for computing/approximating $\delta$ in sub-linear working space on top of the input string. If  $O(n\polylog n)$ time is allowed, their algorithms require $\Theta(n/\polylog n)$ working space, which they proved to be optimal for computing $\delta$ exactly. Our algorithm, on the other hand, computes a multiplicative $(1\pm\epsilon)$-approximation of $\delta$ using working space  polynomial in $\epsilon^{-1}\log n$. 

\section{Preliminaries}\label{sec:preliminaries}
We denote $[n]:=\{1, \ldots, n\}$ for any integer $n$ ($[n]=\emptyset$ for $n\le 0$). For $a \in \mathbb R^+$ and a real number $\eps\in [0, 1]$, we write $[(1\pm\eps)a]$ for the interval $[(1-\eps)a, (1+\eps)a]$. Similarly, we write $[a\pm \eps] $ for the interval $[a-\eps, a+\eps]$.

We assume to be given a string $S$ of length $n > 1$ over an alphabet $\Sigma$ of cardinality $\sigma > 1$. For $k\ge 1$, we define $D_k(S):=\{S[i..i + k - 1]: i\in [n - k + 1]\}$, i.e., the set of all distinct substrings of length $k$ of $S$. Notice that $D_k(S)=\emptyset$ if $k>n$. The \emph{$k$-substring complexity} $d_k(S)$ of $S$ is the cardinality of this set, i.e., $d_k(S) := |D_k(S)|$. 
The \emph{normalized substring complexity} $\delta$ is defined as follows:
\[
    \delta(S) 
    := \max_{k\ge 1} \{ |D_k(S)|/k \}
    = \max_{k\ge 1} \{ d_k(S)/k \}.
\]
We omit the argument from $D_k$, $d_k$, and $\delta$ in case it is clear from the context. 
Here, we also extend this measure to pairs of strings $S$ and $T$. 
Rather than using $\delta(ST)$, we propose the following natural definition that does not take into account artificial length-$k$ substrings crossing the border between $S$ and $T$: 

\[
    \delta(S, T) := \max_{k \ge 1} \{ |D_k(S) \cup D_k(T) | / k \}.
\]

This version also gives mathematically cleaner results (e.g., perfect sub-additivity) and, in any case, differs from $\delta(ST)$ by at most 1. As a consequence, most of our results (read also below) hold also by replacing  $\delta(S, T)$ with $\delta(ST)$. 
The \emph{Normalized Compression Distance} has been defined by Cilibrasi and Vit\'anyi~\cite{cilibrasi2005clustering} as a proxy for the non-computable Normalized Information Distance~\cite{TerwijnTV11}. For two strings $S$ and $T$ and an arbitrary compressibility measure $Z$ (for example, the output size of compression software such as \texttt{gzip} and \texttt{xz}), it is defined as

\[
    \NCD_Z(S, T) := \frac{Z(S, T) - \min \{Z(S), Z(T)\}}{\max \{Z(S), Z(T)\}}.
\]


Given a uniform prime $q = n^{\Theta(1)}$, the \emph{Rabin's fingerprint} \cite{rabin1981fingerprinting} of $S$ is defined as $\rho(S) = \sum_{i=1}^n S[i]\cdot \sigma^{n-i}  \mod q$. 
Collisions between substrings of $S$ through $\rho$ happen with low probability, so the results of our paper hold with high probability. 
We extensively use the fact that the fingerprint of the concatenation of two strings $S_1,S_2$ can be computed in constant time from (i) the fingerprints of $S_1$ and $S_2$ and (ii) $\sigma^{|S_2|}\mod q$ (see \cite{rabin1981fingerprinting}).
Given a set $B \subseteq \Sigma^*$ of strings, we define $\rho(B) = \{\rho(s) \ :\ s\in B\}$.


Given a set $U$, a \emph{count-distinct sketch} $CD(U)$ is a sub-linear-space data structure supporting three main operations: $CD(U).add(x)$, which turns the sketch into $CD(U\cup \{x\})$, $CD(U_1).merge(CD(U_2))$, which turns the sketch into $CD(U_1 \cup U_2)$, and  $CD(U).estimate()$, which returns a $(1\pm \epsilon)$ approximation of $|U|$. In our work, we use the optimal count-distinct sketch of Kane et al.\ \cite{KaneOptimalCountDistinct2010}. Letting $U \subseteq [u]$, this sketch uses $O(\epsilon^{-2} + \log u)$ words of space and computes a $(1\pm \epsilon)$ approximation of $|U|$ with high probability of success. All operations are supported in $O(\log u)$ time\footnote{The authors claim $O(\epsilon^{-2} + \log u)$ \emph{bits} of space and $2/3$ success probability, which can be amplified by taking the median of $\Theta(\log u)$ sketches (thus yielding the bounds we claim above). In our paper, the universe is composed by Rabin's fingerprints and has therefore size $u = n^{\Theta(1)}$.}. 


Assume $S[1]=\$$, where $\$$ is lexicographically smaller than all other alphabet's characters and does not appear anywhere else in $S$. The \emph{Burrows-Wheeler transform} (BWT) of the reverse $S^R$ of $S$ is obtained by sorting lexicographically all suffixes of $S^R$  and then taking, in this order, the character preceding each suffix. For example, if $S =\$babba$, then the sorted suffixes and the BWT of $S^R$ are shown in Table \ref{fig:example BWT}.

\begin{table}[ht!]
    \centering
\begin{tabular}{l|c}
      suffixes of $S^R$ & BWT\\ \hline
         \$ & \textbf b \\
         a b \$ & \textbf b \\
         a b b a b \$ & \emph \$ \\
         b \$  & \textbf  a\\
         b a b \$ & \emph b\\
         b b a b \$ & \textbf a\\
\end{tabular}
\vspace{-5mm}
\caption{Burrows-Wheeler transform $BWT(S^R)$ of the string $S^R = abbab\$$. 
}\label{fig:example BWT}
\vspace{-5mm}
\end{table}

The \emph{LF property} of the BWT states that the $i$-th occurrence of $c\in \Sigma$ in the BWT corresponds 
to the position of the $i$-th suffix starting with $c\in \Sigma$ in Table~\ref{fig:example BWT}. 
The \emph{LF function} is the permutation of $[1,n]$ implementing this observation: for instance, in the above example $BWT.LF(2) = 5$ because character $BWT[2]$ corresponds to the first character ($b$) of the fifth (in lexicographic order) suffix $bab\$$.
We denote with $r$ the number of \emph{equal-letter runs} of the BWT; in the above example, $r=5$ (runs are highlighted in alternating bold/italic). 
We moreover use the following result:

\begin{lemma}[\cite{policriti2018lz77}, Thm.~2]\label{lem:PP-RLBWT}
    Letting $S$ be a string and $r$ be the number of equal-letter runs in $BWT(S^R)$, there exists a data structure of $O(r)$ words storing $BWT(S^R)$
    supporting right-extensions of $S$ (i.e. $BWT(S^R) \rightarrow BWT((Sa)^R)$, for any $a\in \Sigma$) in $O(\log|S|)$ time. Within the same time, the structure supports computing the LF function and retrieving any character of $BWT(S^R)$.
\end{lemma}

\section{Properties of $\delta$ and $\NCD_\delta$}\label{sec:delta and NCD}
We start by proving some properties of $\delta$. 
Proofs of some statements are omitted due to space limitations and can be found in the full version \cite{becker2023sketching}.
The first main property that we show is that $\delta$ is both sub-additive and monotone in the following sense. 
\begin{restatable}{lemma}{subadditivity}\label{lemma: subadditivity}
    For any strings $S$ and $T$, $\max\{\delta(S), \delta(T)\} \le \delta(S, T) \le \delta(S) + \delta(T)$.
\end{restatable}
The proof of the lemma uses the properties of the corresponding maximizers together with the fact that the union is a superset of both its arguments (left inequality) and that the union is of smaller cardinality than the sum of the cardinalities of its arguments (right inequality).
We remark that it is a well-known fact that the monotonicity property holds for the case of concatenation of the two strings~\cite{KociumakaNP23}. 
Using the sub-additivity of $\delta$, we  obtain:
\begin{corollary}\label{cor: ncd range}
    For any strings $S$ and $T$ it holds that $0 \le \NCD_\delta(S, T) \le 1$.
\end{corollary}
To see why this holds, assume, w.l.o.g., that $\max\{\delta(S), \delta(T)\}=\delta(S)$. Then, 
$\NCD_\delta(S, T) = \frac{\delta(S,T) - \delta(T)}{\delta(S)} \ge \frac{\delta(T) - \delta(T)}{\delta(S)} = 0$ and
$\NCD_\delta(S, T) 
\le \frac{\delta(S) + \delta(T) - \delta(T)}{\delta(S)} = 1$.
Ming et al.~\cite{MingSimilarityMetric2004} state that common compressors yield a normalized compression distance between $0$ and $1+\epsilon$, where the $\epsilon$ is due to ``imperfections'' of the compression algorithm. Above we proved that in the case of the normalized substring complexity $\delta$, the corresponding $\epsilon$ is equal to $0$. 

We conclude by showing that a multiplicative approximation of $\delta$ can be used to obtain an additive approximation of the Normalized Compression Distance $\NCD_\delta$.
\begin{restatable}{lemma}{NCDapprox}\label{lem:NCD approx}
    Let $\eps \in (0,1)$, $\eps':=\eps/5$ and let $S$ and $T$ be two strings. Assume that $\tilde \delta(S)$, $\tilde \delta(T)$, and $\tilde \delta(S, T)$ are approximations of $\delta$ in the sense that $\tilde \delta(S)\in [(1\pm\eps')\delta(S)]$, $\tilde \delta(T)\in [(1 \pm \eps')\delta(T)]$, as well as $\tilde \delta(S, T)\in [(1\pm\eps')\delta(S, T)]$. Then 
    \[
        \NCD_{\tilde \delta}(S, T) \in [\NCD_{\delta}(S, T) \pm \epsilon].
    \]
\end{restatable}
We prove this lemma by using the facts that $\tilde \delta$ is a multiplicative approximation of $\delta$, that $\delta$ is sub-additive (see Lemma~\ref{lemma: subadditivity}), and that $\NCD_\delta\in [0, 1]$ (see Corollary~\ref{cor: ncd range}).

\section{A data sketch for estimating $\delta$}\label{sec:sketch}

We introduce our data sketch, then prove that it yields a good approximation of $\delta$.

\begin{definition}[Sketch for $\delta$]\label{def:sketch}
Let $S$ be a string, $A:=\{\lceil \alpha^i \rceil : i\in [\lfloor \log_\alpha n \rfloor ]\}$ be a set of sampled lengths for some real number (sample rate) $\alpha > 1$, and $CD_k = CD(\rho(D_{k}(S)))$, where $CD$ is the count-distinct sketch described in Section~\ref{sec:preliminaries} and $\rho$ is Rabin's hash function. Our data sketch is defined as
$
\kappa(S) = \langle CD_k\ :\ k \in A \rangle.
$
\end{definition}

We define 
$\kappa(S).estimate() = \max\{ CD_k.estimate()/k\ :\ k \in A \}$. When extending the stream $S$ with a new character $a$, yielding string $Sa$, the sketch is updated by calling $CD_k.add(\rho(S[|S|-k+2,|S|]a))$ for all $k\in A$. 
We denote this operation by $\kappa(S).extend(a)$.
Note that, if constant-time random access is available on $S$ and if $\sigma^{k-1}\mod q$ has been pre-computed for all $k\in A$ (in $O(\epsilon^{-1}\log^2 n)$ time), $\rho(S[|S|-k+2,|S|]a)$ can be computed in constant time from $\rho(S[|S|-k+1,|S|])$; see \cite{rabin1981fingerprinting}.
Finally, $\kappa(S_1).merge(\kappa(S_2))$ returns the sketch $\kappa(\{S_1,S_2\}) = \langle CD'_k\ :\ k \in A \rangle$, where $CD'_k = CD^1_k.merge(CD^2_k)$ and $CD^i_k$ is the count-distinct sketch for the (fingerprints of the) length-$k$ substrings of $S_i$, $i\in \{1,2\}$. 
Operation $extend(a)$ is not defined when the sketch represents a set of strings; this is not an issue, since we will call $merge$ only to estimate $\delta(S_1,S_2)$ and $NCD_\delta(S_1,S_2)$.

With the next theorem we show that $\kappa(S).estimate()$ returns a multiplicative $(1\pm\epsilon)$-approximation of $\delta(S)$ (analogous for $\kappa(\{S_1, S_2\}).estimate()$).

\begin{restatable}{lemma}{approxdelta}\label{lemma:approx delta}
    Let $S$ be a string of length $n$. Let $\eps > 0$, $\eps' = \eps/4$, and $\alpha = 1 + \eps'$. Assume that $\tilde d_k\in [(1\pm \eps') d_k(S)]$ for all $k\in A:=\{\lceil \alpha^i \rceil : i\in [\lfloor \log_\alpha n \rfloor ]\}$, then 
    $
        \tilde \delta 
        := \max \{ \tilde d_k/k : k\in A\},
    $
    satisfies 
    $\tilde \delta \in [(1 \pm \eps)\delta]$.
\end{restatable}
We show this theorem by quantifying the impact of two types of errors on the quantity $\delta$. These two types are (1) the error obtained when approximating the values $d_k$ by $\tilde d_k$ and (2) the error due to the restriction of the string's offsets $[n]$ to the set $A$. The error of type~(1) directly implies an error of the same magnitude $(1\pm \eps')$ on $\delta$. We note that this error actually itself has two sources, namely (1.1) errors due to collisions when computing Rabin's fingerprints and (1.2) errors due to the count-distinct sketch when applied to the fingerprints. Both of these errors are accounted for in the assumption $\tilde d_k\in [(1\pm \eps') d_k(S)]$. The error of type~(2) instead is more subtle to analyse -- the main observation here is that $d_{j + 1}\ge d_j - 1$ for every $j$, as every distinct length-$j$ substring other than possibly $S[n - j + 1..n]$ gives at least one distinct length-$(j + 1)$ substring. Now assume that $i\in [n]\setminus A$ and that $a\in A$ is the minimum element of $A$ larger than $i$. Then applying the previous observation iteratively yields $d_{a}\ge d_{i} - \beta$, where $\beta = a - i$. Hence, we can quantify how much $\delta$ gets ``perturbed'' by restricting to the subset $A$ of the string's offsets $[n]$.

From Lemmas \ref{lem:NCD approx} and \ref{lemma:approx delta}, the sketch of Definition \ref{def:sketch} yields a multiplicative $(1\pm \epsilon)$-approximation of $\delta$ and an additive $\epsilon$-approximation of $\NCD_\delta$
if $CD$ is the count-distinct sketch of \cite{KaneOptimalCountDistinct2010} with error rate $\epsilon/20$, and the set $A$ is built with sample rate $\alpha = 1 + \epsilon/20$. 
From \cite{KaneOptimalCountDistinct2010} and since $|A| \in \Theta(\epsilon^{-1}\log n)$, our data sketch uses $\Theta(\epsilon^{-3}\log n + \epsilon^{-1}\log^2 n)$ words of space and supports all operations in time $O(\epsilon^{-1}\log^2 n)$.

Using repeatedly operation $extend$ on our data sketch we immediately obtain:

\begin{theorem}\label{thm: main}
    For any string $S$ of length $n$ supporting random access in time at most $O(\log n),$ and any approximation rate $\epsilon>0$, 
    we can compute a  multiplicative $(1\pm \epsilon)$-approximation of $\delta(S)$ in $O(\epsilon^{-1}n\log^2 n)$ time using  $\Theta(\epsilon^{-3}\log n + \epsilon^{-1}\log^2 n)$ words of working space on top of the string. The result is correct with high probability. 
\end{theorem}


\section{Streaming algorithm}\label{sec:streaming}

We now show how to compute the sketch of Definition \ref{def:sketch} in $O(\sqrt n \log n)$ words of working space (on top of the sketch) with one pass over the streamed input string.

Let $S$ denote the current stream, and $S^R$ be the reversed stream.
We assume that an upper-bound $n$ to the maximum stream length is known before the algorithm starts.
Let $r$ be the number of equal-letter runs in the Burrows-Wheeler transform of $S^R$. By \cite{KempaK20} and by the fact that $\delta$ is invariant under string reversals, it holds $r \leq 8\delta\log^2n$.
Our streaming algorithm works as follows. We keep a sliding window $S[|S|-K+1,|S|]$ of the last $K$ stream characters, for some parameter $K$ to be determined later, and at the same time we keep a dynamic run-length BWT (RLBWT) of $S^R$
that we update by appending the stream's characters using Lemma \ref{lem:PP-RLBWT}.
Before the algorithm starts, in $O(\epsilon^{-1}\log^2 n)$ time we compute $\sigma^{k-1} \mod q$ for all the $|A| \in O(\epsilon^{-1}\log n)$ sampled substring lengths $k\in A$ in our sketch, using fast exponentiation.

Let $k\in A$ be one of the sampled string lengths in our sketch, and let $a$ be a new character arriving on the stream (so that the new stream is $Sa$). In order to update our sketch, we need to compute the fingerprint of the last $k$ stream's characters: $\rho(S[|S|-k+2,|S|]a)$. 
At any stage of the algorithm, we keep the Rabin's fingerprint $\rho(S[2,|S|])$ of the whole stream, excluding character $S[1]=\$$. 
If $|Sa|=k+1$, then $\rho(S[|S|-k+2,|S|]a)$ is equal to the Rabin's fingerprint of the whole stream. 
Otherwise, if $|Sa|>k+1$ then $\rho(S[|S|-k+1,|S|])$ has already been computed in the previous steps and  we can use the formula
$\rho(S[|S|-k+2,|S|]a) = (\rho(S[|S|-k+1,|S|]) - S[|S|-k+1]\cdot \sigma^{k-1})\cdot \sigma + a \mod q$. As a result, updating the fingerprint reduces to extracting character $S[|S|-k+1]$. 
We use the window $S[|S|-K+1,|S|]$ to extract $S[|S|-k+1]$ for any $k\leq K$, and the RLBWT to 
extract $S[|S|-k+1]$ for any $k > K$ 
using a bookmarking technique that we sketch in Figure \ref{fig:example-bookmarking} and we describe in full detail in the full version \cite{becker2023sketching}. This allows us to update the Rabin's fingerprints for all sampled substring lengths $k$ and thus to implement operation $extend(a)$.

We now describe the policy we employ to keep space usage under control.
Let $r'$ be the number of equal-letter runs in the BWT obtained by ignoring (removing) character $\$$. It is easy to see that (i) $r-2 \leq r' \leq r$ and (ii) $r'$ is non-decreasing upon appending characters at the end of the stream.
As soon as $r' \geq 8n(\log^2n)/K$, we discard the RLBWT and keep only the sliding window for the rest of the stream. As a consequence, from this point on we are only able to extract (fingerprints of) length-$k$ substrings with $k\leq K$. However, we show that this is enough: if we discard the RLBWT, then it means that $\delta \geq \frac{r}{8\log^2n} \geq \frac{r'}{8\log^2n} \geq n/K$. Let $\hat k = \argmax_{k\ge 1}\{d_k/k\}$. Then, $\hat k = d_{\hat k}/\delta \leq n/\delta \leq K$ so to compute $\delta$ on the rest of the stream we can focus only on the length-$k$ substrings with $k\leq K$.

The sliding window $S[|S|-K+1,|S|]$ uses $K$ words of space. We discard the RLBWT when $r' \geq 8n(\log^2n)/K$, so (since $r\leq r'+2$) this structure always uses at most $O(r) \subseteq O(n(\log^2n)/K)$ words.
As a consequence, in total we use $O(K + n(\log^2n)/K)$ words of space, which is optimized asymptotically when $K=\sqrt n \log n$; then, our algorithm uses at most $O(\sqrt n \log n)$ words of space.

\begin{figure}
    \centering
\begin{minipage}[c]{0.49\textwidth}
\centering
\begin{tabular}{l|c}
     suffixes of $S^R$ & BWT\\\hline
     \$ & \$
\end{tabular}
\end{minipage}
\begin{minipage}[c]{0.49\textwidth}
\centering
\begin{tabular}{l|c}
     suffixes of $S^R$ & BWT\\\hline
     \$ & b\\
     b\$ & \$
\end{tabular}
\end{minipage}

\vspace{10pt}

\begin{minipage}[c]{0.49\textwidth}
\centering
\begin{tabular}{l|cl}
     suffixes of $S^R$ & BWT\\\hline
     \$ & b & $j=1$\\
     ab\$ & \$ &\\
     b\$ & a &\\
\end{tabular}
\end{minipage}
\begin{minipage}[c]{0.49\textwidth}
\centering
\begin{tabular}{l|cl}
     suffixes of $S^R$ & BWT\\\hline
     \$ & b&\\
     aab\$ & \$ &\\
     ab\$ & a &\\
     b\$ & a& $j=4$\\
\end{tabular}
\end{minipage}
\caption{\footnotesize Example showing how the BWT (of the reversed stream) is updated upon character right-extensions of the stream, and how the bookmark $j$ corresponding to window length $k=2$ is initialized and updated. \textbf{Top left}: empty stream ($S = \$$). \textbf{Top right}: a new character $b$ arrives on the stream ($S = \$b$): in the BWT, \$ is replaced by $b$ and a new \$ is inserted in the position corresponding to the lexicographic rank $i=2$ of the new suffix $b\$$. 
Position $i$ is computed in $O(\log|S|)$ time using the algorithm described in \cite[Thm. 2]{policriti2018lz77}.
\textbf{Bottom left}: a new character $a$ arrives on the stream ($S = \$ba$): in the BWT, \$ is replaced by $a$ and a new \$ is inserted in the position corresponding to the lexicographic rank $i=2$ of the new suffix $ab\$$. Since the stream length is equal to $k+1=3$, we initialize the bookmark $j \leftarrow BWT.LF(i) = BWT.LF(2) = 1$. Note that $BWT[j]=b$ indeed contains character $S[|S|-k+1] = b$. \textbf{Bottom right}: a new character $a$ arrives on the stream ($S = \$baa$): in the BWT, \$ is replaced by $a$ and a new \$ is inserted in the position corresponding to the lexicographic rank $i=2$ of the new suffix $aab\$$. Since $1= j < i = 2$ (\$ is inserted after position $j$), $j=1$ is not modified (otherwise, it would have been incremented by 1). Finally, we update $j$ by advancing it  by one position in the text: $j \leftarrow BWT.LF(j) = BWT.LF(1) = 4$. Note that $BWT[j]=a$ indeed contains character $S[|S|-k+1] = a$. Importantly, the data structure of \cite[Thm. 2]{policriti2018lz77} uses always a space proportional to the number $r$ of equal-letter runs of the BWT.}
\vspace{-3mm}
\label{fig:example-bookmarking}
\end{figure}

We keep one bookmark (a position in the BWT) for every sampled length $k\in A$, so our bookmarking technique does not affect the asymptotic working space if $\epsilon \geq n^{-1/2}$ (i.e. $|A| \leq \sqrt n \log n$). Updating each bookmark and extracting $S[|S|-k+1]$ from the RLBWT take $O(\log n)$ time by Lemma \ref{lem:PP-RLBWT}. This running time is absorbed by operation $merge()$ on the count-distinct sketches, see Section \ref{sec:sketch}. We obtain:

\begin{theorem}
    Given an upper-bound $n$ to the stream's length, we can compute the sketch of Definition \ref{def:sketch} in $O(\sqrt n\log n)$ words of working space and $O(\epsilon^{-1}\log^2 n)$ worst-case delay per stream character, for any approximation factor  $\epsilon \geq n^{-1/2}$.
\end{theorem}
\vspace{-7mm}

\section{Implementation and experiments}\label{sec:exp}
\vspace{-1mm}

We implemented a parallel version of our streaming algorithm in {\tt C++}.\footnote{ \url{https://github.com/regindex/substring-complexity}} 
We ran experiments on a server with Intel(R) Xeon(R) W-2245 CPU @ 3.90GHz with 16 threads and 128GB of RAM running Ubuntu 18.04 LTS 64-bit.
Our complete experimental results are reported in \cite{becker2023sketching}. 
We used the repetitive real Pizza\&Chilli dataset (P\&C)\footnote{\url{https://pizzachili.dcc.uchile.cl/repcorpus/real/}}, large Canterbury corpus\footnote{\url{http://corpus.canterbury.ac.nz/resources/large.tar.gz}}, and datasets from AF Project\footnote{\url{https://afproject.org}}.
We computed the relative error of our approximation $\tilde \delta$ with respect to $\delta$ for different sampling densities (i.e. parameter $\alpha$ of Definition \ref{def:sketch}). 
With the sparsest (less precise) sampling scheme (option {\tt -p 1}), $\tilde\delta$ always differed from $\delta$ by up to 5\% and the average throughput was of 174 MB per minute using up to 16 threads (option {\tt -t 0}). For efficiency reasons, the RLBWT is disabled by default: in practice this does not affect precision, since $\hat k=\argmax_k d_k/k$ was always extremely small ($\hat k\le 100$ in all datasets), meaning that the RLBWT is never required. We also computed ${\tilde \delta}$ on a big dataset of 189GB long reads of Rana Muscosa\footnote{\url{https://trace.ncbi.nlm.nih.gov/Traces/?view=run_browser&acc=SRR11606868}}. Our software finished the computation in 15:31 hours with a throughput of 203MB per minute using only about 5MB of internal memory.

\noindent\textbf{Experiments on repetitiveness measures.} We studied the effectiveness of ${\tilde \delta}$ as a repetitive measure. We compared it to exact $\delta$, to the number of runs of the BWT $r$, to the number of phrases of the LZ77 parse $z$, and to the output of two popular compressors, {\tt xz} and {\tt 7z}. For each dataset in the repetitive P\&C corpus, we computed these five measures for prefixes of increasing length. We observe that ${\tilde \delta}$ not only follows closely the values of $\delta$, but it also mirrors the trend of the other four measures. This  suggests experimentally that ${\tilde \delta}$ computed by our streaming algorithm is a good indicator of repetitiveness and compressibility.

\begin{figure*}[ht!]
    \centering
 	\includegraphics[width=0.44\textwidth, trim={5.5mm 5.5mm 5.0mm 5.5mm}, clip]{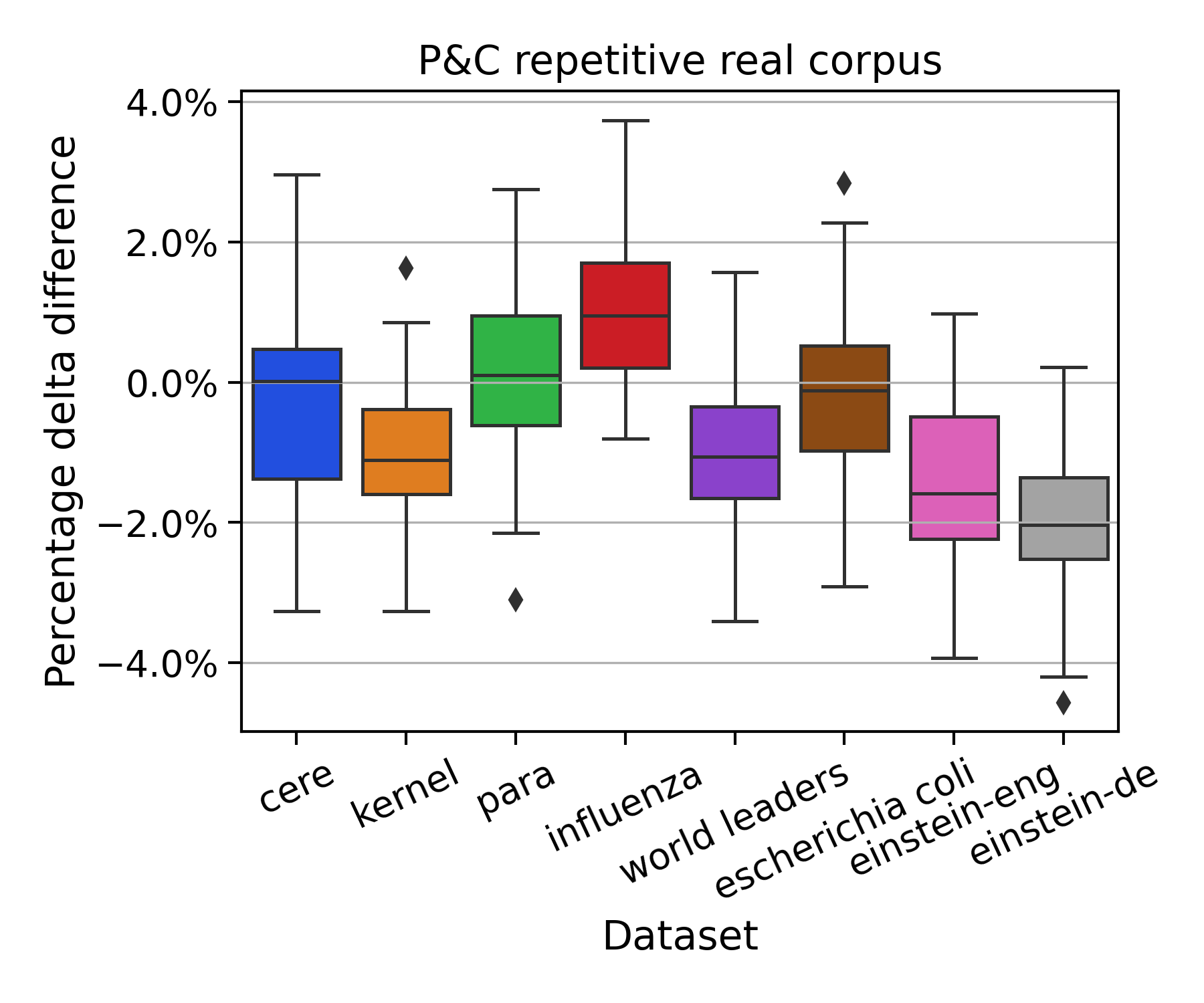}
    \hspace{3mm}
    \includegraphics[width=0.44\textwidth, trim={5.5mm 5.5mm 5.0mm 5.5mm}, clip]{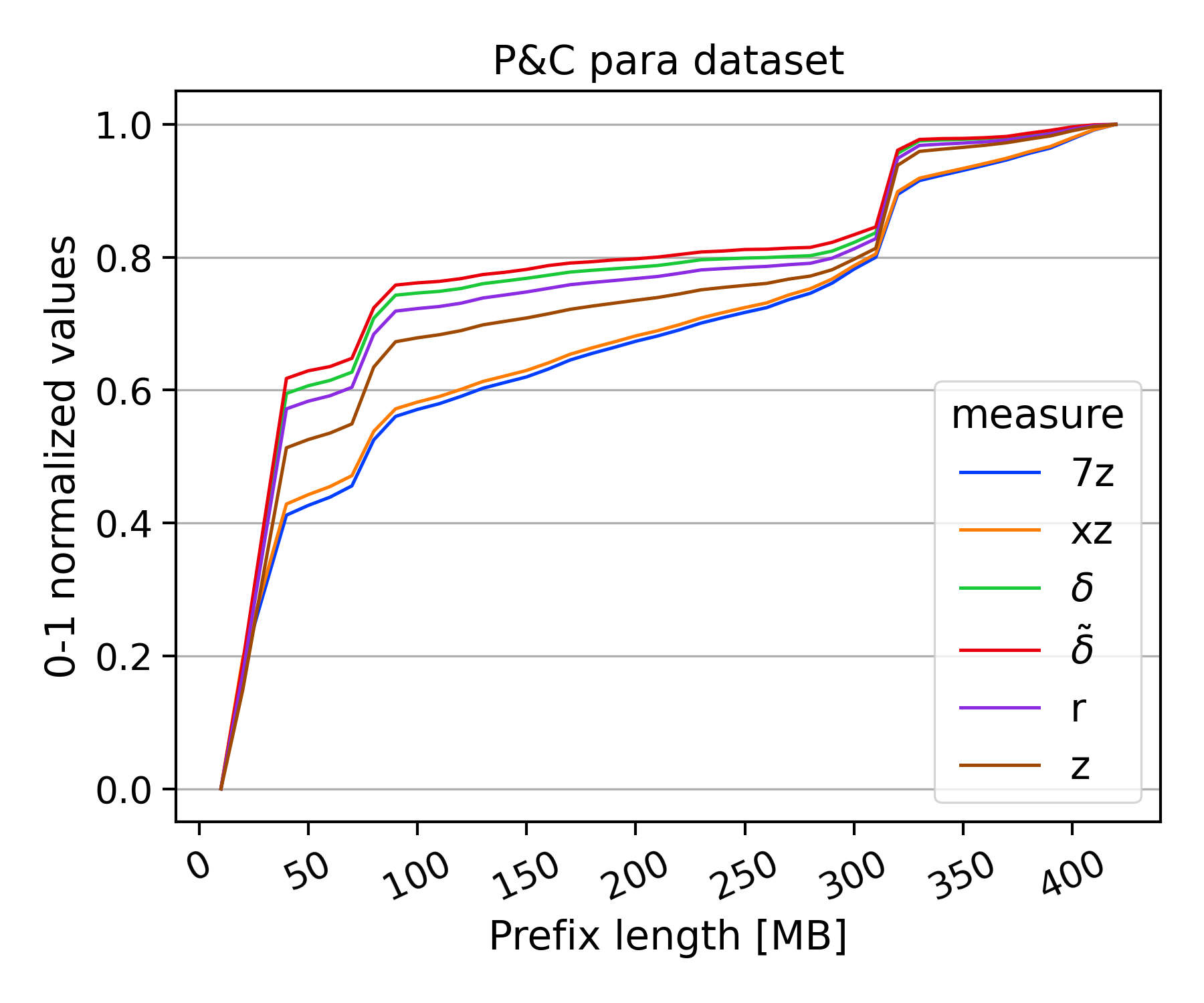}
    \caption{
    From left to right: $\tilde \delta$ error distribution on P\&C repetitive corpus, lineplot showing five repetitiveness measures (normalized to $[0,1]$) computed on increasing prefixes of \textit{para}. 
    }
    \label{fig:exp1}
    \vspace{-3mm}
\end{figure*}

\noindent\textbf{Experiments on phylogeny reconstruction.} We verified that NCD based on the compression software $\mathtt{xz}$, on $\delta$, and on $\tilde\delta$ yield similar phylogenetic trees with the \texttt{Gene-trees} dataset from AF Project; the average normalized Robinson-Foulds distances ranged from 0.1 to 0.3, indicating that the reconstructed trees were very similar. We also measured the running time to compute all-pair NCDs on 29 sequences of average length $\sim$81k. This process took only 3 minutes for $\tilde\delta$ and 24 minutes for the exact $\delta$, while for $\mathtt{xz}$ it required 42 minutes.




\Section{References}\vspace{-0.4cm}
\bibliographystyle{IEEEbib}
\bibliography{refs}

\newpage
\appendix

\section{Missing Proofs}
We start with the following two easy observations that we use in our proofs later on.
\begin{observation}\label{obs: lower 2}
    If $S\neq a^n$ for $a\in \Sigma$, then $\delta(S) \ge d_1/1 \ge 2/1 = 2$.
\end{observation}
Note that if at least two distinct letters appear, it follows that $d_1 > 1$ and thus $\delta \ge  d_1/1 \ge 2$.
Since it is easy to recognize the case $S=a^n$ for some $a\in \Sigma$ in constant space and constant delay per character, from now on we assume w.l.o.g. that $\delta\ge 2$. 

We continue with the following simple observation that we will use in the proof of Lemma~\ref{lemma:approx delta}.
\begin{observation}\label{obs: hat k upper bound}
    It holds that $\hat k = \argmax_{k\ge 1}\{d_k(S)/k\}\le n/2$.
\end{observation}
\begin{proof} 
    Assume that $\hat k > n/2$. Then it is immediate that $d_{\hat k} \le n - \hat k + 1 < n/2$ (as the right most character in the substring can be at index at most $n$). We now obtain that $d_{\hat k} / \hat k < 1$, contradicting Observation~\ref{obs: lower 2}.
\end{proof}

\subadditivity*
\begin{proof}
    Let $k_{S, T}$, $k_S$, and $k_T$ be such that $\delta(S, T) = |D_{k_{S, T}}(S) \cup D_{k_{S, T}}(T) | / k_{S, T}$, $\delta(S) = |D_{k_S}(S)|/k_S$, and $\delta(T) = |D_{k_T}(T)|/k_T$. Let, w.l.o.g., $\delta(S)=\max\{\delta(S), \delta(T)\}$. Then, 
    \begin{align*}
        \delta(S, T) 
        &= \frac{|D_{k_{S, T}}(S) \cup D_{k_{S, T}}(T) |}{k_{S, T}}
        \ge \frac{|D_{k_{S}}(S) \cup D_{k_{S}}(T) |}{k_{S}}
        \ge \frac{|D_{k_{S}}(S)|}{k_{S}}
        = \delta(S),
    \end{align*}
    where the second inequality uses the fact that $k_S$ is the maximizer for $S$.
    For the second claim, 
    \begin{align*}
        \delta(S, T) 
        &= \frac{|D_{k_{S, T}}(S) \cup D_{k_{S, T}}(T) |}{k_{S, T}}
        \le \frac{|D_{k_{S, T}}(S)|}{k_{S, T}} + \frac{|D_{k_{S, T}}(T) |}{k_{S, T}}\\
        &\le \frac{|D_{k_{S}}(S)|}{k_{S}} + \frac{|D_{k_{T}}(T) |}{k_{T}}
        = \delta(S) + \delta(T),
    \end{align*}
    where the second inequality uses the fact that $k_S$ and $k_T$ are the respective maximizers for $S$ and $T$.
\end{proof}

\NCDapprox*
\begin{proof}
    We start with the lower bound. Using the definition of $\NCD_{\tilde \delta}(S, T)$, we obtain
    \begin{align*}
        \NCD_{\tilde \delta}(S, T)
        &\ge \frac{(1 - \eps')\cdot \delta(S, T) - (1 + \eps') \min \{\delta(S), \delta(T)\}}{(1 + \eps')\max \{\delta(S), \delta(T)\}}\\
        & = \frac{1}{1 + \eps'} \cdot \NCD_{\delta}(S, T) - \frac{\eps'}{1 + \eps'} \cdot \frac{\delta(S, T) + \min \{\delta(S), \delta(T)\}}{\max \{\delta(S), \delta(T)\}}\\
        & = \NCD_{\delta}(S, T) - \frac{\eps'}{1 + \eps'} \cdot \Big(\frac{\delta(S, T) + \min \{\delta(S), \delta(T)\}}{\max \{\delta(S), \delta(T)\}} + \NCD_{\delta}(S, T)\Big)\\
        & \ge \NCD_{\delta}(S, T) - \frac{4\eps'}{1 + \eps'}\\
        &\ge \NCD_{\delta}(S, T) - \eps
    \end{align*}
    using the sub-additivity of $\delta$ from Lemma~\ref{lemma: subadditivity}, the fact that $\NCD_{\delta}(S, T)\le 1$ according to Corollary~\ref{cor: ncd range}, and the definition of $\eps'$.
    Similarly, now for the upper bound, we obtain
    \begin{align*}
        \NCD_{\tilde \delta}(S, T)
        &\le \frac{(1 + \eps')\cdot \delta(S, T) - (1 - \eps') \min \{\delta(S), \delta(T)\}}{(1 - \eps')\max \{\delta(S), \delta(T)\}}\\
        & = \frac{1}{1 - \eps'} \cdot \NCD_{\delta}(S, T) + \frac{\eps'}{1 - \eps'} \cdot \frac{\delta(S, T) + \min \{\delta(S), \delta(T)\}}{\max \{\delta(S), \delta(T)\}}\\
        & = \NCD_{\delta}(S, T) + \frac{\eps'}{1 - \eps'} \cdot \Big(\frac{\delta(S, T) + \min \{\delta(S), \delta(T)\}}{\max \{\delta(S), \delta(T)\}} + \NCD_{\delta}(S, T)\Big)\\
        & \le \NCD_{\delta}(S, T) + \frac{4\eps'}{1 - \eps'}\\
        &\le \NCD_{\delta}(S, T) + \eps
    \end{align*}
    again using the sub-additivity of $\delta$ from Lemma~\ref{lemma: subadditivity}, the fact that $\NCD_{\delta}(S, T)\le 1$ according to Corollary~\ref{cor: ncd range}, the definition of $\eps'$ and the assumption that $\eps<1$.
\end{proof}

\approxdelta*
\begin{proof}
    We first observe that $\lceil \alpha^{\lfloor \log_\alpha n\rfloor} \rceil \le n$. To see this, assume otherwise, i.e., that $\alpha^{\lfloor \log_\alpha n\rfloor} = n + x$ for some $x>0$. Then $x =  \alpha^{\lfloor \log_\alpha n\rfloor} - n \le \alpha^{ \log_\alpha n} - n = 0$, contradicting the assumption that $x>0$. It follows that $A\subseteq [n]$. Now, for the upper bound notice that $\tilde \delta \le \max \{ (1+\eps') d_k / k : k\in A\} \le (1 + \eps') \cdot \delta\le (1 + \eps) \cdot \delta$.
    
    For the lower bound, let $\hat k\in [n]$ be such that $\delta=d_{\hat k}/\hat k$ and let $i$ be such that $\lceil\alpha^{i-1}\rceil \le \hat k \le \lceil \alpha^i \rceil$. Notice that obviously $\lceil\alpha^{i-1}\rceil \in A$, but also $\lceil \alpha^i \rceil\in A$ as $\alpha^i \le \alpha \hat k\le \alpha n/2\le n$ by Observation~\ref{obs: hat k upper bound}.
    We now distinguish two cases: (1) $\lceil \alpha^i \rceil = \lceil \alpha^{i - 1} \rceil + 1$ and (2) $\lceil \alpha^i \rceil \ge \lceil \alpha^{i-1} \rceil + 2$. In case (1), we get that $\hat k\in \{\lceil\alpha^{i-1} \rceil, \lceil\alpha^{i} \rceil\} \subseteq A$ and consequently $\tilde \delta \ge \max \{ (1 - \eps') d_k / k : k\in A\} = (1 - \eps') \cdot \delta \ge (1 - \eps) \cdot \delta$. In case (2), it holds that 
    \begin{align}\label{formula: alphai lower bound}
        \alpha^{i - 1} \cdot \eps'
        = \alpha^i  - \alpha^{i - 1}
        \ge \lceil \alpha^i \rceil - 1 - \alpha^{i - 1}
        \ge \lceil \alpha^{i - 1} \rceil + 1 - \alpha^{i - 1}
        \ge 1.
    \end{align}
    Now let $\beta := \lceil \alpha^i \rceil - \hat k$. We note that $d_{j + 1}\ge d_j - 1$ for every $j$, as every distinct length-$j$ substring other than possibly $S[n - j + 1..n]$ gives at least one distinct length-$j + 1$ substring. Applying the same observation iteratively yields $d_{\lceil \alpha^i \rceil}\ge d_{\hat k} - \beta$. Hence
    \begin{align*}
        \tilde \delta 
        \ge (1 - \eps') \cdot \frac{d_{\lceil \alpha^i \rceil}}{\lceil \alpha^i \rceil}
        \ge (1 - \eps') \cdot \frac{d_{\hat k} - \beta}{\hat k + \beta}
        = (1 - \eps') \delta \cdot \frac{1 - \frac{\beta}{d_{\hat k}}}{1 + \frac{\beta}{{\hat k}}}
        \ge (1 - \eps') \delta \cdot \frac{1 - \frac{\beta}{2{\hat k}}}{1 + \frac{\beta}{{\hat k}}},
    \end{align*}
    where we used that $\delta(S) = d_{\hat k}/ \hat k \ge 2$ in the last step. We can now upper bound $\beta$ by $\lceil \alpha^i \rceil - \lceil \alpha^{i-1} \rceil \le \alpha^i + 1 - \alpha^{i-1} = \alpha^{i-1} \cdot \eps' + 1 \le \hat k \eps' + 1$. This yields 
    \[
        \hat \delta 
        \ge (1 - \eps')\delta \cdot \frac{1 - \frac{\eps'}{2}  - \frac{1}{2 \hat k}}{1 + \eps' + \frac{1}{\hat k}}
        \ge (1 - \eps')\delta \cdot \frac{1 - \eps'}{1 + 2 \eps'}
        \ge (1 - \eps) \cdot \delta,
    \]
    where the second inequality uses that $\hat k\ge 1/\eps'$ following from~\eqref{formula: alphai lower bound} and the last inequality uses the definition of $\eps'$.
\end{proof}

\section{Details on Bookmarking the RLBWT}\label{app:bookmarking}

We show how to extract $S[|S|-k+1]$ from the RLBWT, for any of the sampled lengths $k$. See also the example in Figure \ref{fig:example-bookmarking}.
We show how to initialize and update (upon character extensions of the stream) an index (bookmark) $j$ such that $BWT[j] = S[|S|-k+1]$. This allows us retrieving $S[|S|-k+1]$ in $O(\log |S|) \subseteq O(\log n)$ time with a random access operation $BWT[j]$ on the RLBWT data structure.

We first discuss how to initialize the bookmark $j$ as soon as the stream's length becomes $S = k+1$  (before that, the window of the last $k$ characters is not completely filled).
The initialization works by setting $j = BWT.LF(i)$, where $i$ is the position such that $BWT[i] = \$$. Since the LF mapping on the BWT of the reversed stream corresponds to advancing one position in the stream, it is easy to see that, after this operation, it holds $BWT[j] = S[|S|-k+1]$. See Figure \ref{fig:example-bookmarking} for an example.

Suppose we are storing the bookmark $j$ such that $BWT[j] = S[|S|-k+1]$.
We now show how to update $j$ when a new character $a$ arrives; let  $S' = Sa$ be the updated stream. Our goal is to modify $j$ so that $BWT[j] = S'[|S'|-k+1]$ holds. The observation is that, upon the extension of the stream by one character $a$, the algorithm of \cite{policriti2018lz77} modifies the $BWT$ as follows: letting $i$ being the index such that $BWT[i]=\$$, the algorithm (1) replaces $BWT[i] \leftarrow a$, and (2) inserts \$ in the position $i'$ corresponding to the lexicographic rank of the new reversed stream $(Sa)^R$ (position $i'$ is computed in $O(\log|S|)$ time using basic operations on the RLBWT, see \cite{policriti2018lz77} and Example \ref{fig:example-bookmarking}): the new BWT becomes $BWT \leftarrow BWT[1,i'-1]\cdot \$ \cdot BWT[i'+1,|S|]$. If $j < i'$ (i.e. \$ is inserted after position $j$), then after these modification we have that $BWT[j] = S'[|S'|-k]$; if, on the other hand, $j \geq i'$ (i.e. \$ is inserted before position $j$), then we increment $j$ as $j\leftarrow j+1$, and $BWT[j] = S'[|S'|-k]$ holds also in this case. 
Finally, we need to ``advance'' $j$ by one position on the stream; this operation corresponds to one LF mapping step on the BWT: $j \leftarrow BWT.LF(j)$ ($O(\log |S|)$ time). After these operations, we finally have that $BWT[j] = S'[|S'|-k+1]$.

\section{Detailed Experimental Results}\label{app:experiments}

\subsection{Estimation of $d_k$.} 
As mentioned above, there are two types of errors in the computation of the approximation $\tilde \delta$ of $\delta$:  (1) the error obtained when approximating the values $d_k$ by $\tilde d_k$ and (2) the error due to the restriction of the string's offsets $[n]$ to the ``sampled set'' $A$. The error of type~(1) itself has two sources, namely (1.1) errors due to collisions when computing fingerprints with Rabin's hash function and (1.2) errors due to the count-distinct sketch when applied to the fingerprints. We experimentally evaluated the error of type~(1.1) and (1.2) as follows.
For the Pizza\&Chili repetitive corpus, we compute the exact values of $d_k$ and their estimated values $\tilde d_k$ for $k\in\{2^i:0\le i \le 7\}$. 
We observe that the error (1.1) caused by collisions in Rabin's fingerprint were negligible; the error in the ratio of the distinct number of fingerprints and the actual number of distinct substrings was less than 0.01\%. The error due to the count-distinct sketch (1.2) was dependent on its parameter: the number of registers used for estimation. It is worth noting that the number of registers does not affect the time complexity when updating sketches, but only affects the space usage by a constant factor (and the time to compute the actual estimation at the end, which is negligible). When more than $2^{14}$ registers were used for count-distinct sketches, the maximum relative error was observed to be below $2\%$, and the average error on $d_k$ was below $0.5\%$; see Table~\ref{tbl:approx:dk}.

\begin{table}[h]
\centering

\begin{tabular}{c|cccc}\hline
The number of registers & $2^{10}$ & $2^{12}$ & $2^{14}$ & $2^{16}$ \\\hline
Maximum & 0.0641 & 0.0313	& 0.0186 & 0.0077 \\
Average & 0.0204 & 0.0089 & 0.0041 & 0.0017 \\\hline
\end{tabular}
\caption{The relative error measured in estimating $d_k$} \label{tbl:approx:dk}
\end{table}

\subsection{Experiments on phylogenetic tree reconstruction.} 

To show similar behavior of $\NCD_{\mathtt{xz}}$, $\NCD_\delta$, and $\NCD_{\tilde\delta}$, we conducted experiments on phylogenetic tree reconstruction using \texttt{Gene-Trees} dataset from AF project\footnote{\url{https://afproject.org/}}. It contains 11 groups of sequences (651 sequences in total) where each group yields a tree. We constructed 11 phylogenetic trees (i.e. one tree for each group) for each of the NCD measures, and compare the constructed trees by measuring the normalized Robinson-Foulds (nRF) distance, a widely-used distance measure for this purpose. The distance tends to 0 as trees become similar, and tends to 1 when comparing with a random tree. The average nRF between $\NCD_{\mathtt{xz}}$ and $\NCD_\delta$ was measured as 0.2, indicating similar trees were reconstructed. The average nRF between $\NCD_{\delta}$ and $\NCD_{\tilde\delta}$ ranges from 0.115 to 0.250 depending on the parameters. For ease of interpretation of these values, we depict two similar phylogenetic trees with nRF=0.194 in Figure~\ref{fig:appendix:phylogeny}, which is an actual example of reconstructed trees using NCD with $\tilde\delta$ and $\mathtt{xz}$.

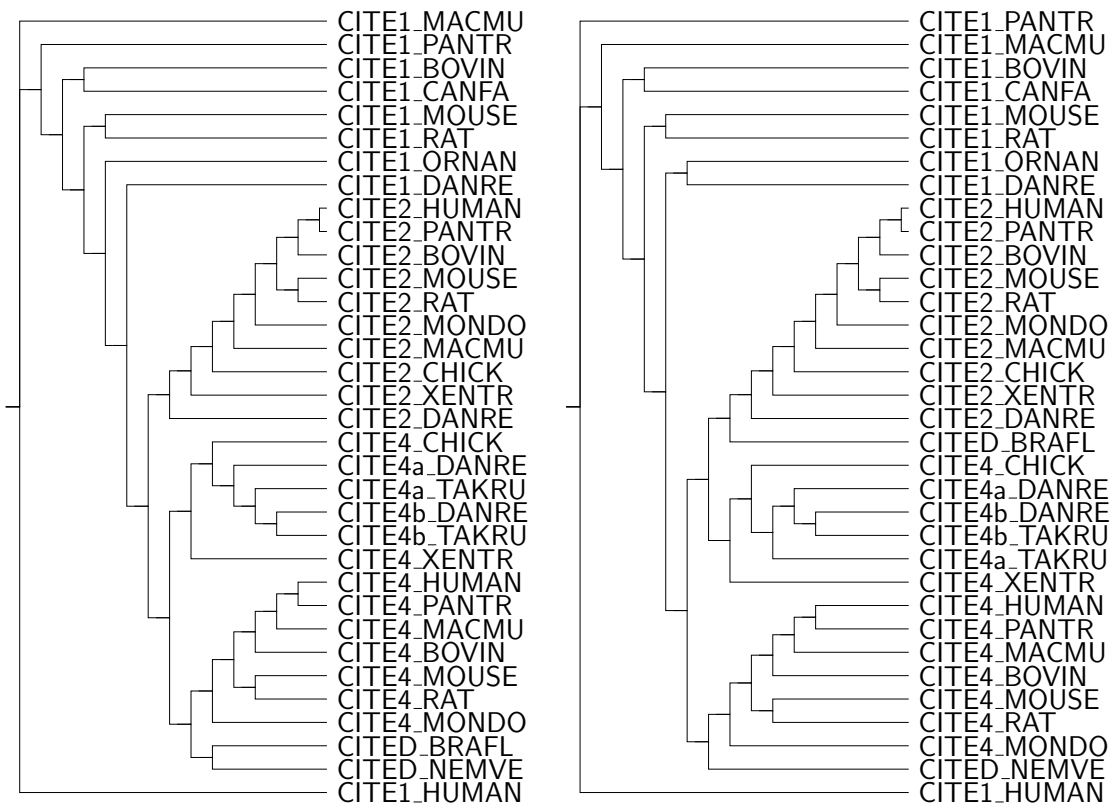
\begin{figure}
    \centering
    \scalebox{.9}{
        \begin{forest}
  forked edges,
  /tikz/every label/.append style={font=\sffamily},
  before typesetting nodes={
    delay={
      where content={}{coordinate}{},
    },
    where n children=0{tier=terminus, label/.process={Ow{content}{right:#1}}, content=}{},
  },
  for tree={
    grow'=0,
    s sep'+=2pt,
    l=1pt,
    l sep=9pt,
  },
[[CITE1\_MACMU][[CITE1\_PANTR][[[CITE1\_BOVIN][CITE1\_CANFA]][[[CITE1\_MOUSE][CITE1\_RAT]][[CITE1\_ORNAN][[CITE1\_DANRE][[[[[[[[[[CITE2\_HUMAN][CITE2\_PANTR]][CITE2\_BOVIN]][[CITE2\_MOUSE][CITE2\_RAT]]][CITE2\_MONDO]][CITE2\_MACMU]][CITE2\_CHICK]][CITE2\_XENTR]][CITE2\_DANRE]][[[[CITE4\_CHICK][[CITE4a\_DANRE][[CITE4a\_TAKRU][[CITE4b\_DANRE][CITE4b\_TAKRU]]]]][CITE4\_XENTR]][[[[[[[CITE4\_HUMAN][CITE4\_PANTR]][CITE4\_MACMU]][CITE4\_BOVIN]][[CITE4\_MOUSE][CITE4\_RAT]]][CITE4\_MONDO]][[CITED\_BRAFL][CITED\_NEMVE]]]]]]]]]][CITE1\_HUMAN]]
\end{forest}\ \ 
        \begin{forest}
  forked edges,
  /tikz/every label/.append style={font=\sffamily},
  before typesetting nodes={
    delay={
      where content={}{coordinate}{},
    },
    where n children=0{tier=terminus, label/.process={Ow{content}{right:#1}}, content=}{},
  },
  for tree={
    grow'=0,
    s sep'+=2pt,
    l=1pt,
    l sep=9pt,
  },
[[CITE1\_PANTR][[CITE1\_MACMU][[[CITE1\_BOVIN][CITE1\_CANFA]][[[CITE1\_MOUSE][CITE1\_RAT]][[[CITE1\_ORNAN][CITE1\_DANRE]]
[
[[[[[[[[[[[CITE2\_HUMAN][CITE2\_PANTR]][CITE2\_BOVIN]][[CITE2\_MOUSE][CITE2\_RAT]]][CITE2\_MONDO]][CITE2\_MACMU]][CITE2\_CHICK]][CITE2\_XENTR]][CITE2\_DANRE]][CITED\_BRAFL]][[[CITE4\_CHICK][[[CITE4a\_DANRE][[CITE4b\_DANRE][CITE4b\_TAKRU]]][CITE4a\_TAKRU]]][CITE4\_XENTR]]]
[[[[[[[CITE4\_HUMAN][CITE4\_PANTR]][CITE4\_MACMU]][CITE4\_BOVIN]][[CITE4\_MOUSE][CITE4\_RAT]]][CITE4\_MONDO]][CITED\_NEMVE]]
]
]]]][CITE1\_HUMAN]]
\end{forest}
    }
    \caption{Two similar phylogenetic trees constructed using normalized compression distance (NCD) with estimated normalized substring complexity $\tilde\delta$ (left) and a popular compression software $\mathtt{xz}$ (right). Normalized Robinson-Foulds distance is 0.194.}
    \label{fig:appendix:phylogeny}
\end{figure}

\paragraph{Running Time.} To construct a phylogenetic tree from a sequence set, we usually need to compute all-pair distances. When sequences are long, computing NCDs can be quite costly because we need to compress concatenated sequences for all pairs of sequence in the input set. On the other hand, our sketching can be more efficient because we only need to compute sketches for each sequence, then merging sketches can be done very quickly compared to processing the entire sequences all over again. To demonstrate this, we measure the running time for computing all-pair NCDs on \texttt{E.coli}
 dataset from AF Project that consists of 29 sequences of average length 81,588. Computing all-pair NCDs with \texttt{xz} and the exact $\delta$ took about 42 and 24 minutes. On the other hand, our sketching method took only 3 minutes.

\paragraph{Funding}

Ruben Becker, Davide Cenzato, Sung-Hwan Kim, Bojana Kodric, and Nicola Prezza are funded by the European Union (ERC, REGINDEX, 101039208). Views and opinions expressed are however those of the author(s) only and do not necessarily reflect those of the European Union or the European Research Council. Neither the European Union nor the granting authority can be held responsible for them.

\end{document}